\documentclass[aps,pra,twocolumn,superscriptaddress,
preprintnumbers,amsmath,amssymb,floatfix]{revtex4-2}

\usepackage{graphicx}
\usepackage{subfigure}
\usepackage{amsthm}
\usepackage{tensor}
\usepackage{color}
\usepackage[all]{xy}
\usepackage{tikz}
\usepackage{dsfont}
\usepackage{times,txfonts}
\usetikzlibrary{positioning}
\usepackage{braket}
\newtheorem{Thm}{Theorem}

\theoremstyle{definition}

\newcommand{\Tr}{\mathop{\mathrm{Tr}}\nolimits}

\begin{document}

\title{Contrasting Effects of Control on Fidelity and Fidelity Deviation in Controlled Teleportation}

\author{Jeonghyeon Shin} 
\affiliation{
Center for Quantum Information, Korea Institute of Science and Technology (KIST), Seoul 02792, Republic of Korea
}
\affiliation{
Department of Mathematics and Research Institute for Basic Sciences,
Kyung Hee University, Seoul 02447, Republic of Korea}

\author{Minjin Choi}
\email{mathcmj89@gmail.com}
\affiliation{
Center for Quantum Information R\&D, Korea Institute of Science and Technology Information~(KISTI), Daejeon 34141, Republic of Korea
}

\date{\today}

\begin{abstract}
Teleportation performance is commonly characterized by the average teleportation fidelity, while its variation over input states provides additional information captured by the fidelity deviation.
In controlled teleportation, the controller's measurement introduces an additional source of fidelity variation through its measurement outcomes.
We investigate the fidelity deviation in controlled teleportation with three-qubit pure states.
We derive a lower bound on the fidelity deviation and show that it is attainable together with the maximal average teleportation fidelity.
Among the measurements attaining the maximal average fidelity, however, the fidelity deviation can vary depending on the controller's measurement.
We further examine the effect of controller assistance by comparing controlled teleportation with direct teleportation using the reduced state.
Unlike the maximal average fidelity, which cannot decrease with controller assistance, the minimal fidelity deviation is not necessarily reduced by controller assistance.
In particular, for any W-class pure state, controller assistance cannot reduce the fidelity deviation.
These results reveal contrasting effects of control on teleportation fidelity and its deviation in controlled teleportation.
\end{abstract}
\maketitle

\maketitle

\section{Introduction}
\label{sec: Introduction}

Quantum teleportation is one of the fundamental protocols of quantum information processing, enabling the transmission of an unknown quantum state between distant parties through shared quantum resources and classical communication~\cite{Bennett1993, Bouwmeester1997}.
Its performance is commonly characterized by the average teleportation fidelity, which quantifies the fidelity between the input and teleported states averaged over all possible inputs~\cite{Popescu1994, Horodecki1999, Badziag2000}.
The average fidelity, however, does not capture how the teleportation fidelity varies from one input state to another.
Such variations can be relevant even when two teleportation processes exhibit the same average fidelity, as their fidelities for individual input states may be distributed differently around the average.
To characterize this aspect of teleportation, the fidelity deviation was introduced as the standard deviation of the teleportation fidelity over input states~\cite{Bang2018}.
Subsequent work established a general expression for the fidelity deviation corresponding to optimal teleportation with an arbitrary two-qubit resource state~\cite{Ghosal2020}.
Recent studies have also explored teleportation performance beyond the average fidelity from the perspectives of fidelity deviation and fidelity distribution~\cite{Cho2026, Bussandri2026}.
The average fidelity and fidelity deviation thus characterize complementary features of the fidelity distribution, namely its average value and fluctuations across input states.

Controlled teleportation involves three parties sharing a quantum state, where one party acts as a controller and performs a local measurement to assist teleportation between the other two parties~\cite{Karlsson1998, Gao2008, Lee2005, Lee2007}.
Depending on the measurement outcome, the remaining two parties generally share different conditional resource states, thereby affecting the achievable teleportation fidelity.
The controller can therefore influence the teleportation performance through the choice of measurement.
In particular, optimizing over the controller's measurements determines the maximal average teleportation fidelity achievable with controller assistance~\cite{Lee2005, Lee2007, Choi2024}.
The effect of control has also been quantified by comparing the teleportation fidelities achievable with and without the controller's assistance, leading to the notion of control power~\cite{Li2014, Li2015, Jeong2016}.

Beyond its effect on the average teleportation fidelity, the controller's measurement can also influence the fidelity deviation.
In controlled teleportation, the fidelity deviation involves not only the variation over input states for each conditional resource state but also the variation associated with the controller's measurement outcomes.
This motivates examining how the fidelity deviation depends on the controller's measurement.
A further question is how the effect of controller assistance on the fidelity deviation can be characterized by comparing controlled teleportation with direct teleportation using the reduced state, analogous to the comparison underlying the control power.

In this work, we analyze the fidelity deviation in controlled teleportation with three-qubit pure states.
We characterize the fidelity deviation associated with the controller's measurement and derive a lower bound that is attainable together with the maximal average teleportation fidelity.
This analysis also clarifies the effect of the controller's measurement choice on the fidelity deviation for the same maximal average fidelity.
We further compare controlled teleportation with direct teleportation using the reduced state to examine the effect of controller assistance on the fidelity deviation.
In contrast to the fidelity enhancement characterized by the control power, controller assistance does not necessarily reduce the fidelity deviation.
For W-class pure states, we establish that the fidelity deviation cannot be decreased through controller assistance, whereas such a reduction is possible for states outside the W class.
These results show that the fidelity deviation provides a distinct perspective on the effect of control beyond that captured by the average teleportation fidelity.

The remainder of this paper is organized as follows.
Sec.~\ref{sec: telecapa} reviews the average teleportation fidelity and fidelity deviation for two-qubit resource states, and introduces the controlled teleportation setting together with its maximal average fidelity.
Sec.~\ref{sec:total_variance} analyzes the fidelity deviation under the controller's measurement, and establishes its lower bound and attainability for three-qubit pure states.
Sec.~\ref{sec:effect_control} examines the effect of controller assistance on the fidelity deviation, with particular attention to W-class pure states.
Finally, Sec.~\ref{sec: conclusion} concludes the paper.

\section{Teleportation capability in controlled teleportation} 
\label{sec: telecapa}

Quantum teleportation is a protocol that enables two parties, Alice and Bob, sharing a two-qubit state $\rho_{AB}$, to transmit an unknown qubit state using only local operations and classical communication~\cite{Bennett1993}.
For a given teleportation scheme $\Lambda_{\rho_{AB}}$ using $\rho_{AB}$, the average teleportation fidelity is defined as~\cite{Popescu1994}
\begin{equation}
\label{eq:def_fidel}
F\left(\Lambda_{\rho_{AB}}\right)=\int d\xi\bra{\xi}\Lambda_{\rho_{AB}}(|\xi\rangle\langle\xi|)\ket{\xi},
\end{equation}
where the integral runs over the uniform distribution of all one-qubit pure states.
Introducing the fully entangled fraction of $\rho_{AB}$~\cite{Bennett1996},
\begin{equation}
\label{eq:fraction}
f(\rho_{AB})=\max_{\ket{e}}\bra{e}\rho_{AB}\ket{e},
\end{equation}
where the maximum is taken over maximally entangled two-qubit states $\ket{e}$,
it has been shown that the maximal average teleportation fidelity achievable within the standard teleportation protocol is given by~\cite{Horodecki1999, Badziag2000}
\begin{equation}
\label{eq:two-qubit_fidel}
F(\rho_{AB})=\frac{2f(\rho_{AB})+1}{3}.
\end{equation}
We denote this maximal average fidelity by $F\left(\rho_{AB}\right)$ throughout this paper.

Controlled teleportation extends this protocol to three parties sharing a three-qubit state $\rho_{ABC}$, one of whom acts as the controller.
The controller first performs an orthogonal measurement on his or her own qubit, and depending on the outcome, the remaining two parties carry out the standard teleportation protocol using the resulting two-qubit state.
For distinct $i,j,k \in \{A,B,C\}$, let $\rho_{ij}^{M,t}$ denote the two-qubit state shared by parties $i$ and $j$ after party $k$ measures his or her part in an orthonormal basis $M=\{\ket{u_0},\ket{u_1}\}$ and obtains outcome $t \in \{0, 1\}$~\cite{Karlsson1998}.
The maximal average fidelity between $i$ and $j$ achievable with controller assistance is then given by~\cite{Lee2005, Lee2007}
\begin{equation}
\label{eq:tele_fidel}
F_{ij}(\rho_{ABC})=\frac{2f_{ij}(\rho_{ABC})+1}{3},
\end{equation}
where
\begin{equation}
\label{eq:tele_frac}
f_{ij}(\rho_{ABC})=\max_{M}\sum_{t=0}^{1}
\bra{u_t}\rho_{k}\ket{u_t}f\left(\rho_{ij}^{M, t}\right),
\end{equation}
and the maximum is taken over the controller's choice of measurement basis.

The maximal average fidelity of a two-qubit pure state is closely related to the amount of entanglement it contains.
For a two-qubit pure state $\ket{\phi}_{AB}$, the Schmidt decomposition gives $\ket{\phi}_{AB}=\sqrt{a}\ket{u_{0}v_{0}}+\sqrt{1-a}\ket{u_{1}v_{1}}$, where $0 \le a \le 1$, and $\{\ket{u_{0}}, \ket{u_{1}}\}$ and $\{\ket{v_{0}}, \ket{v_{1}}\}$ are orthonormal bases.
It follows from a direct calculation that $F(\ket{\phi}_{AB})=\frac{2}{3}+\frac{2}{3}\sqrt{a(1-a)}$.
The concurrence, introduced by Wootters for two-qubit states~\cite{Hill1997, Wootters1998} and generalized to arbitrary-dimensional bipartite pure states by Rungta \textit{et al.}~\cite{Rungta2001}, is given for a pure state $\ket{\psi}_{P_{1}P_{2}}$ by
\begin{equation}
\label{eq:def_concurrence}
C_{P_{1}|P_{2}}\left(\ket{\psi}_{P_{1}P_{2}}\right)=\sqrt{2\left(1-\mathrm{Tr}\left(\varrho_{P_{1}}^{2}\right)\right)},
\end{equation}
where $\varrho_{P_{1}}$ is the reduced density operator of $\ket{\psi}_{P_{1}P_{2}}$.
Therefore, $C_{A|B}(\ket{\phi}_{AB})=2\sqrt{a(1-a)}$, which yields
\begin{equation}
\label{eq:pure_fidelity_concurrence}
F(\ket{\phi}_{AB})=\frac{2}{3}+\frac{1}{3}C_{A|B}(\ket{\phi}_{AB}).
\end{equation}
Hence, $\ket{\phi}_{AB}$ is separable if and only if $F(\ket{\phi}_{AB})=2/3$. 
Moreover, for two-qubit pure states, greater entanglement corresponds to a higher maximal average fidelity.

The maximal average fidelity of a three-qubit pure state is similarly related to its entanglement structure.
For a three-qubit pure state $\ket{\phi}_{ABC}$ and distinct $i,j,k \in \{A,B,C\}$, let $\rho_{ij}$ denote the two-qubit state obtained by tracing out qubit $k$.
The maximal average fidelity then satisfies the relation~\cite{Lee2005}
\begin{equation}
\label{eq:partial_tangle_fidelity}
F_{ij}(\ket{\phi}_{ABC})=\frac{2}{3}+\frac{1}{3}\tau_{ij}(\ket{\phi}_{ABC}),
\end{equation}
where
\begin{equation}
\label{eq:partial_tangle}
\tau_{ij}(\ket{\phi}_{ABC})=\sqrt{C_{i|j}(\rho_{ij})^{2}+\tau_{ABC}(\ket{\phi}_{ABC})},
\end{equation}
with the three-tangle $\tau_{ABC}$ given by~\cite{Coffman2000}
\begin{equation}
\tau_{ABC}(\ket{\phi}_{ABC})=C_{i|jk}(\ket{\phi}_{ABC})^{2}-C_{i|j}(\rho_{ij})^{2}-C_{i|k}(\rho_{ik})^{2}.
\end{equation}
The concurrence of the mixed state $\rho_{ij}$ can be calculated by~\cite{Hill1997, Wootters1998}
\begin{equation}
\label{eq:concurrence_mixed}
C_{i|j}(\rho_{ij})=\max\{0, \lambda_{1}-\lambda_{2}-\lambda_{3}-\lambda_{4}\},
\end{equation}
where $\lambda_{1},\ldots, \lambda_{4}$ are the square roots of the eigenvalues of $\rho_{ij}\tilde{\rho}_{ij}$ in decreasing order. 
Here, $\tilde{\rho}_{ij}=(\sigma_y\otimes \sigma_y)\rho^*_{ij}(\sigma_y\otimes \sigma_y)$ with the Pauli $Y$ operator $\sigma_{y}$ and the complex conjugate $\rho_{ij}^{*}$ of $\rho_{ij}$.

The quantity $\tau_{ij}$ can be directly expressed in terms of the concurrence of assistance~\cite{Laustsen2003,Gour2005}.
For a two-qubit state $\rho$, the concurrence of assistance is defined as
\begin{equation}
\label{eq:coa_definition}
C_{a}(\rho)=\max_{\{p_{\mu},\ket{\eta_{\mu}}\}}\sum_\mu p_{\mu} C(\ket{\eta_\mu}),
\end{equation}
where the maximum is over all pure-state decompositions $\rho=\sum_\mu p_\mu|\eta_\mu\rangle\langle\eta_\mu|$.
When $\rho$ is regarded as the reduced state of a larger pure state, $C_{a}(\rho)$ quantifies the maximum average concurrence attainable with assistance from the remaining subsystem.
For a reduced two-qubit state $\rho_{ij}$ of a three-qubit pure state $\ket{\phi}_{ABC}$, the concurrence of assistance satisfies~\cite{Chi2010}
\begin{equation}
\label{eq:coa_tangle}
C_a(\rho_{ij})^2=C_{i|j}(\rho_{ij})^2+\tau_{ABC}(\ket{\phi}_{ABC}).
\end{equation}
Comparing Eqs.~\eqref{eq:partial_tangle} and \eqref{eq:coa_tangle} gives $\tau_{ij}(\ket{\phi}_{ABC})=C_a(\rho_{ij})$, and hence Eq.~\eqref{eq:partial_tangle_fidelity} can equivalently be written as
\begin{equation}
\label{eq:ct_fidelity_coa}
F_{ij}(|\phi\rangle_{ABC})=\frac{2+C_a(\rho_{ij})}{3}.
\end{equation}
Thus, for three-qubit pure states, a larger average entanglement attainable with assistance corresponds to a higher maximal average fidelity.

\section{Fidelity deviation in controlled teleportation} 
\label{sec:total_variance}

While the maximal average fidelity $F(\rho_{AB})$ characterizes the average performance of teleportation, it does not describe how uniformly different input states are teleported.
To quantify the fluctuation of teleportation fidelity over input states, the fidelity deviation was introduced as~\cite{Bang2018}
\begin{equation}
\Delta(\Lambda_{\rho_{AB}})=\sqrt{\int d\xi \left[\bra{\xi}\Lambda_{\rho_{AB}}(|\xi\rangle\langle\xi|)\ket{\xi} - F(\Lambda_{\rho_{AB}})\right]^{2}},
\end{equation}
for a given teleportation scheme $\Lambda_{\rho_{AB}}$.
A smaller fidelity deviation indicates less variation in teleportation fidelity among different input states, with zero deviation corresponding to equal fidelity for all input states.
For a resource state $\rho_{AB}$, we consider the fidelity deviation corresponding to the standard teleportation protocol achieving the maximal average fidelity $F(\rho_{AB})$, and denote it by $\Delta(\rho_{AB})$.
For a two-qubit pure state $\ket{\phi}_{AB}$, this fidelity deviation can be expressed in terms of concurrence as~\cite{Ghosal2020}
\begin{equation}
\label{eq:pure_concurrence_deviation}
\Delta(\ket{\phi}_{AB})=\frac{1}{3\sqrt{5}}[1-C_{A|B}(\ket{\phi}_{AB})].
\end{equation}
Together with Eq.~\eqref{eq:pure_fidelity_concurrence}, this shows that, for two-qubit pure states, greater entanglement corresponds to a higher maximal average fidelity and a smaller fidelity deviation.
In particular, the fidelity deviation vanishes if and only if the state is maximally entangled.

In controlled teleportation, the controller's measurement introduces an additional source of fidelity fluctuation, since different measurement outcomes generally leave the remaining two parties with different resource states.
Therefore, in addition to the variation over input states characterized by the fidelity deviation, the variation among the controller's measurement outcomes should also be taken into account.
Let $\rho_{ABC}$ be a three-qubit state, with party $k \in \{A, B, C\}$ acting as the controller, and consider an orthogonal projective measurement $M$ on qubit $k$. 
For each outcome $t \in \{0, 1\}$ with probability $p_t$, we denote, for simplicity, the maximal average fidelity and fidelity deviation of the resulting two-qubit state by $F_{t}$ and $\Delta_{t}$, respectively.
By the law of total variance, the total variance of the teleportation fidelity associated with $M$ is given by
\begin{equation}
\label{eq:CT_variance}
\Delta^{2}_{ij}(M;\rho_{ABC})=\mathbb{E}\left[\Delta_{t}^{2}\right]+\operatorname{Var}(F_{t}),
\end{equation}
where $\mathbb{E}$ and $\operatorname{Var}$ denote the expectation and variance, respectively, over the controller's measurement outcomes with probabilities $p_{t}$.
The first term in Eq.~\eqref{eq:CT_variance} represents the average variance over input states within each measurement outcome, whereas the second represents the variance of the average fidelities across different outcomes.
We refer to $\Delta_{ij}(M;\rho_{ABC})$ as the fidelity deviation of controlled teleportation associated with $M$, and define the minimal fidelity deviation over the controller's measurements as
\begin{equation}
\label{eq:minimal_fidelity_deviation}
\Delta_{ij}(\rho_{ABC})=\min_{M}\Delta_{ij}(M;\rho_{ABC}).
\end{equation}

For a three-qubit pure state $\ket{\phi}_{ABC}$, each measurement outcome leaves the remaining two parties in a pure state.
Denoting the concurrence of the resulting state for outcome $t$ by $C_{t}$, Eqs.~\eqref{eq:pure_fidelity_concurrence} and \eqref{eq:pure_concurrence_deviation} give $F_{t}=(2+C_{t})/3$ and $\Delta_{t}=(1-C_{t})/(3\sqrt{5})$.
Substituting these relations into Eq.~\eqref{eq:CT_variance} yields
\begin{align}
\label{eq:CT_variance_concurrence}
\Delta^{2}_{ij}(M;\ket{\phi}_{ABC}) &=\frac{1}{45}\mathbb{E}\left[(1-C_{t})^{2}\right]+\frac{1}{9}\operatorname{Var}(C_{t}) \nonumber \\
&=\frac{1}{45}(1-\mathbb{E}[C_{t}])^{2}+\frac{2}{15}\operatorname{Var}(C_{t}).
\end{align}
Eq.~\eqref{eq:CT_variance_concurrence} yields a lower bound on the squared fidelity deviation that holds for any controller measurement $M$.
It follows from $\mathbb{E}[C_{t}] \le C_{a}(\rho_{ij}) \le 1$ and $\operatorname{Var}(C_{t}) \ge 0$, where $\rho_{ij}=\Tr_{k}(|\phi\rangle\langle\phi|_{ABC})$, that
\begin{equation}
\label{eq:variance_lower_bound}
\Delta^{2}_{ij}(M;\ket{\phi}_{ABC}) \ge \frac{1}{45}\left[1-C_a(\rho_{ij})\right]^2.
\end{equation}
Equality holds if and only if
\begin{equation}
\mathbb{E}[C_{t}]=C_a(\rho_{ij})\quad\text{and}\quad\operatorname{Var}(C_{t})=0,
\end{equation}
or equivalently, $C_{t}=C_a(\rho_{ij})$ for every outcome with nonzero probability.
Therefore, a measurement saturating Eq.~\eqref{eq:variance_lower_bound} simultaneously achieves the maximal average fidelity $F_{ij}$ and attains the minimal fidelity deviation $\Delta_{ij}$. 
We next show that such a measurement always exists for any three-qubit pure state.

\begin{Thm}
\label{res:thm1}
For any three-qubit pure state $\ket{\phi}_{ABC}$ with party $k \in \{A,B,C\}$ acting as the controller, there exists an orthogonal projective measurement $M^{\star}$ such that the concurrence $C_{t}$ of the resulting two-qubit state satisfies
\begin{equation}
C_t=C_a(\rho_{ij}),
\end{equation}
for every measurement outcome $t$ with nonzero probability, where $\rho_{ij}=\Tr_{k}(|\phi\rangle\langle\phi|_{ABC})$.
Consequently, $M^{\star}$ attains the minimal fidelity deviation $\Delta_{ij}$ by saturating the lower bound in Eq.~\eqref{eq:variance_lower_bound}, while simultaneously achieving the maximal average fidelity $F_{ij}$.
\end{Thm}
\begin{proof}
It suffices to consider party $A$ as the controller, since the other cases follow by permutation of the subsystems.
Up to local unitaries, an arbitrary three-qubit pure state can be written in the Ac\'in canonical form~\cite{Acin2000, Acin2001}
\begin{equation}
\label{eq:acin_form}
\ket{\phi}_{ABC}=\lambda_{0}\ket{000}+\lambda_{1}e^{i\theta}\ket{100}+\lambda_{2}\ket{101}+\lambda_{3}\ket{110}+\lambda_{4}\ket{111},
\end{equation}
where $\lambda_m\geq0$, $\sum_{m=0}^{4}\lambda_m^2=1$, and $0\leq\theta\leq\pi$.
For this state, the concurrence of assistance of $\rho_{BC}$ is given by~\cite{Chi2010}
\begin{equation}
\label{eq:Coa_rhoBC}
C_a(\rho_{BC})=2\sqrt{\lambda_0^{2}\lambda_4^{2}+|\lambda_1\lambda_4e^{i\theta}-\lambda_2\lambda_3|^{2}}.
\end{equation}
Consider an orthonormal measurement basis $\{\ket{u_0},\ket{u_1}\}$ on qubit $A$, with $\bra{u_t}=u_{t0}\bra{0}+u_{t1}\bra{1}$, and denote the probability of outcome $t$ by $p_{t}$.
Define
\begin{equation}
X_t=\lambda_0\lambda_4u_{t0}+\left(\lambda_1\lambda_4e^{i\theta}-\lambda_2\lambda_3\right)u_{t1}.
\end{equation}
A straightforward calculation gives $p_tC_t=2|u_{t1}||X_t|$.
Hence, $C_{t}=C_{a}(\rho_{BC})$ is obtained whenever
\begin{equation}
\label{eq:average_optimal_condition}
p_t=|u_{t1}|^2 \quad\text{and}\quad |X_t|=\frac{C_a(\rho_{BC})}{2}|u_{t1}|
\end{equation}
for every outcome with nonzero probability.

We now show that an orthogonal measurement satisfying Eq.~\eqref{eq:average_optimal_condition} always exists.
Parameterize the measurement basis as
\begin{align}
|u_0\rangle&=\cos\frac{\gamma}{2}|0\rangle+e^{i\phi}\sin\frac{\gamma}{2}|1\rangle,\notag\\
|u_1\rangle&=-e^{-i\phi}\sin\frac{\gamma}{2}|0\rangle+\cos\frac{\gamma}{2}|1\rangle.
\label{eq:measurement_parameterization}
\end{align}
Choosing $\phi=\pi/2$, a direct calculation for outcome $t=0$ gives
\begin{align}
\label{eq:p0 condition}
&p_0-|u_{01}|^2=\lambda_0\left(\lambda_0\cos\gamma+\lambda_1\sin\theta\sin\gamma\right), \nonumber \\
&|X_0|^2-\frac{C_a(\rho_{BC})^2}{4}|u_{01}|^2=\lambda_0\lambda_4^2\left(\lambda_0\cos\gamma+\lambda_1\sin\theta\sin\gamma\right).
\end{align}
Thus, both conditions in Eq.~\eqref{eq:average_optimal_condition} for $t=0$ are satisfied by choosing $\gamma$ such that $\lambda_0\cos\gamma+\lambda_1\sin\theta\sin\gamma=0$.
One may choose $\gamma$ satisfying $\tan\gamma=-\lambda_0/(\lambda_{1}\sin\theta)$ if $\lambda_{1}\sin\theta \neq 0$, and $\gamma=\pi/2$ otherwise.

For this choice, completeness gives $p_1=1-p_0=|u_{11}|^2$. 
Moreover, from the definition of $X_{t}$, the orthonormality of the measurement basis, and Eq.~\eqref{eq:Coa_rhoBC},
\begin{equation}
\sum_{t=0}^{1}|X_t|^2=\frac{C_a(\rho_{BC})^{2}}{4}.
\end{equation}
Together with Eq.~\eqref{eq:p0 condition}, this gives
\begin{equation}
|X_1|^2=\frac{C_a(\rho_{BC})^2}{4}|u_{11}|^2.
\end{equation}
Hence, the conditions in Eq.~\eqref{eq:average_optimal_condition} hold for both outcomes, yielding $C_t=C_a(\rho_{BC})$ for every outcome with nonzero probability.
\end{proof}

Theorem~\ref{res:thm1} shows that the maximal average fidelity can always be attained simultaneously with the minimal fidelity deviation.
Consequently, for any three-qubit pure state, the minimal fidelity deviation in Eq.~\eqref{eq:minimal_fidelity_deviation} is given by
\begin{equation}
\label{eq:ct_deviation_coa}
\Delta_{ij}(\ket{\phi}_{ABC})=\frac{1}{3\sqrt{5}}\left[1-C_a(\rho_{ij})\right].
\end{equation}
Combining Eqs.~\eqref{eq:ct_fidelity_coa} and \eqref{eq:ct_deviation_coa}, the minimal fidelity deviation is directly related to the maximal average fidelity as
\begin{equation}
\Delta_{ij}(\ket{\phi}_{ABC})=\frac{1}{\sqrt{5}}\left[1-F_{ij}(\ket{\phi}_{ABC})\right].
\end{equation}
Thus, the optimal values of the average fidelity and fidelity deviation are linearly related for three-qubit pure states.

However, while the maximal average fidelity only requires $\mathbb{E}[C_{t}]=C_a(\rho_{ij})$, attaining the minimal fidelity deviation additionally requires $\operatorname{Var}(C_{t})=0$.
The fidelity deviation therefore provides an additional criterion for distinguishing controller measurements that yield the same maximal average fidelity.
For instance, consider the state
\begin{equation}
\ket{\psi}_{ABC}=\frac{1}{2}(\ket{000}+\ket{101}+\ket{110}+\ket{111}).
\end{equation}
Let $M_{1}$ and $M_{2}$ be measurements on qubit $A$ in the bases $\{(\ket{0}\pm i\ket{1})/\sqrt{2}\}$ and $\{\cos\alpha\ket{0}-\sin\alpha\ket{1}, \sin\alpha\ket{0}+\cos\alpha\ket{1}\}$, respectively, where $\alpha=3\pi/8$.
Both measurements attain the maximal average fidelity $F_{BC}(\ket{\psi}_{ABC})=(4+\sqrt{2})/6 \approx 0.9024$.
However, their fidelity deviations are different and are given by
\begin{align}
&\Delta_{BC}(M_{1};\ket{\psi}_{ABC})=\Delta_{BC}(\ket{\psi}_{ABC})=\frac{2-\sqrt{2}}{6\sqrt{5}}, \nonumber \\
&\Delta_{BC}(M_{2};\ket{\psi}_{ABC})=\sqrt{\frac{27-14\sqrt{2}}{630}}.
\end{align}
Thus, controller measurements yielding the same maximal average fidelity can exhibit different fidelity fluctuations.

This distinction is particularly transparent for the W-class pure states.
Up to local unitaries, a three-qubit W-class pure state can be written as~\cite{Chi2010}
\begin{equation}
\label{eq:w_canonical}
|\phi_W\rangle_{ABC}=\lambda_0|000\rangle+\lambda_1|100\rangle+\lambda_2|101\rangle+\lambda_3|110\rangle,
\end{equation}
where $\lambda_m\geq0$ and $\sum_{m=0}^{3}\lambda_m^2=1$.
For party $A$ acting as the controller, the vanishing three-tangle of the W class gives
\begin{equation}
C_{a}(\rho_{BC})=C_{B|C}(\rho_{BC})=2\lambda_2\lambda_3.
\end{equation}
On the other hand, it follows from the expression for $p_{t}C_{t}$ in the proof of Theorem~\ref{res:thm1} that, for any orthogonal projective measurement on qubit $A$,
\begin{equation}
\mathbb{E}[C_{t}]=2\lambda_2\lambda_3.
\end{equation}
Hence, all such measurements achieve the same maximal average fidelity.
In contrast, the fidelity deviation generally depends on the controller's measurement.
Figure~\ref{fig1} illustrates this measurement dependence for the standard W state, $\ket{W}_{ABC}=(\ket{001}+\ket{010}+\ket{100})/\sqrt{3}$, with the controller's measurement basis parameterized as $\{\cos\beta\ket{0}+\sin\beta\ket{1}, \sin\beta\ket{0}-\cos\beta\ket{1}\}$.
Although all measurements attain the same maximal average fidelity $F_{BC}(\ket{W}_{ABC})=8/9$, the corresponding fidelity deviation varies substantially with the measurement basis.
Therefore, fidelity deviation can reveal aspects of controlled teleportation that are not captured by the average fidelity alone.

\begin{figure}[t]
\centering
\includegraphics[width=\columnwidth]{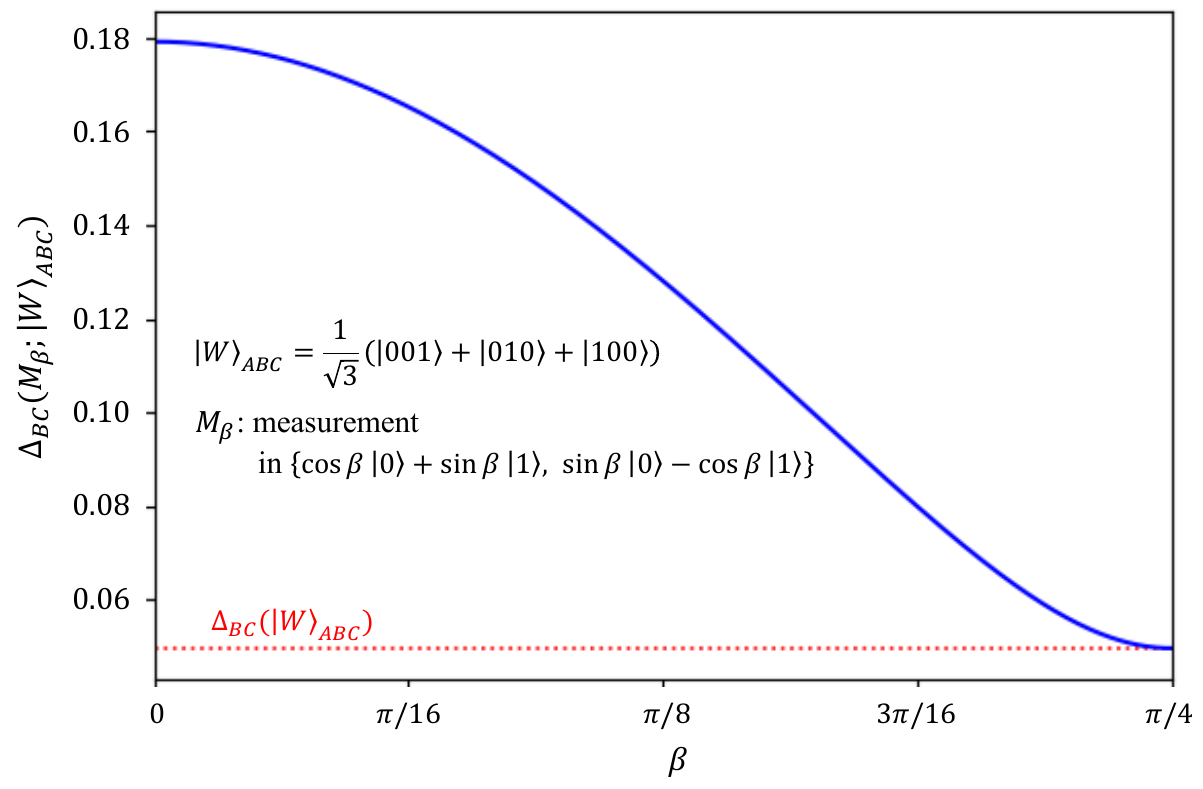}
\caption{
Fidelity deviation for the standard W state as a function of the controller's measurement parameter $\beta$.
All measurements $M_{\beta}$ attain the same maximal average teleportation fidelity $F_{BC}(\ket{W}_{ABC})=8/9$, whereas the fidelity deviation depends on $\beta$ and reaches its minimum $\Delta_{BC}(\ket{W}_{ABC})$ at $\beta=\pi/4$.
}
\label{fig1}
\end{figure}

\section{Effect of control on fidelity deviation} 
\label{sec:effect_control}

The effect of control in controlled teleportation can be characterized by comparing the teleportation performance with and without the controller's assistance.
Recall that the maximal average fidelity satisfies
\begin{equation}
\label{eq:diff_fidel}
F_{ij}(\ket{\phi}_{ABC}) \ge F(\rho_{ij}),
\end{equation}
since, for any orthogonal measurement $M$ on the controller's qubit, $f(\rho_{ij}) \le \sum_{t}p_{t}f(\rho_{ij}^{M,t})$.
The difference between the two fidelities has been introduced as the control power, quantifying the improvement in the average teleportation fidelity enabled by the controller~\cite{Li2014, Li2015, Jeong2016}.

It is then natural to ask whether an analogous relation holds for fidelity deviation.
Specifically, one may expect that the minimal fidelity deviation satisfies
\begin{equation}
\Delta_{ij}(\ket{\phi}_{ABC}) \le \Delta(\rho_{ij}).
\end{equation}
Unlike the case of average teleportation fidelity, however, this relation does not hold in general.
In particular, for three-qubit W-class pure states, the opposite inequality always holds, as established in the following theorem.

\begin{Thm}
\label{res:thm2}
For any three-qubit W-class pure state $\ket{\phi_W}_{ABC}$ with party $k \in \{A,B,C\}$ acting as the controller, the fidelity deviation obtained by directly using the reduced state $\rho_{ij}$ of the two parties without the controller's assistance does not exceed the minimal fidelity deviation achievable with the controller's assistance.
That is,
\begin{equation}
\label{eq:w_reduced_bound}
\Delta(\rho_{ij})\le\Delta_{ij}(\ket{\phi_W}_{ABC})=\frac{1-C(\rho_{ij})}{3\sqrt{5}}.
\end{equation}
\end{Thm}

\begin{proof}
Without loss of generality, we take party $A$ to be the controller, as the other cases follow by permutation of the subsystems.
For the canonical form of a $W$-class pure state in Eq.~\eqref{eq:w_canonical}, tracing out subsystem $A$ gives
\begin{align}
\rho_{BC} = \lambda_0^2|00\rangle\langle00|+|v\rangle\langle v|,
\label{eq:w_reduced_state}
\end{align}
where $|v\rangle = \lambda_1|00\rangle+\lambda_2|01\rangle+\lambda_3|10\rangle$.
Its correlation matrix, defined by $[T_{BC}]_{mn}=\Tr[\rho_{BC}(\sigma_m\otimes\sigma_n)]$ with Pauli operators $\{\sigma_{m}\}_{m=1}^{3}$, is
\begin{equation}
    T_{BC} = 
    \begin{pmatrix} 2\lambda_2\lambda_3 & 0 & 2\lambda_1\lambda_3\\0 & 2\lambda_2\lambda_3 & 0\\2\lambda_1\lambda_2 & 0 & \lambda_0^2+\lambda_1^2-\lambda_2^2-\lambda_3^2 \end{pmatrix}
\label{eq:w_reduced_correlation}
\end{equation}
Let $\mu=2\lambda_2\lambda_3$.
The singular values of $T_{BC}$ are then $\mu,\nu_+,\nu_-$, where $\nu_+ \ge \nu_- \ge 0$ are the singular values of
\begin{equation}
M =
\begin{pmatrix}\mu & 2\lambda_1\lambda_3\\2\lambda_1\lambda_2 & 1-2(\lambda_2^2+\lambda_3^2)\end{pmatrix}.
\end{equation}
These singular values satisfy
\begin{equation}
\nu_-\leq\mu\leq\nu_+.
\label{eq:singular_order}
\end{equation}
Indeed, $\nu_+\nu_-=|\det M|=\mu|(2\lambda_0^2-1)|$ and
\begin{equation}
\quad \nu_++\nu_-=\sqrt{\Tr(M^{\mathsf T}M)+2|\det M|},
\end{equation}
where $\Tr(M^{\mathsf T}M)=1+\mu^2-4\lambda_0^2(\lambda_2^2+\lambda_3^2)$.
Using $\sum_{m=0}^3\lambda_m^2=1$, we obtain
\begin{equation}
(\nu_++\nu_-)^2
-\left(\mu+|2\lambda_0^2-1|\right)^2 = 4\lambda_0^2\lambda_1^2 \geq 0.
\end{equation}
Hence, $\nu_++\nu_- \ge \mu+|2\lambda_0^2-1|$, and therefore
\begin{align}
(\mu-\nu_+)(\mu-\nu_-)=\mu\left(\mu+|(2\lambda_0^2-1)|-(\nu_{+}+\nu_{-})\right) \le 0,
\end{align}
which proves Eq.~\eqref{eq:singular_order}.

\begin{figure*}[t]
\centering
\includegraphics[width=\textwidth]{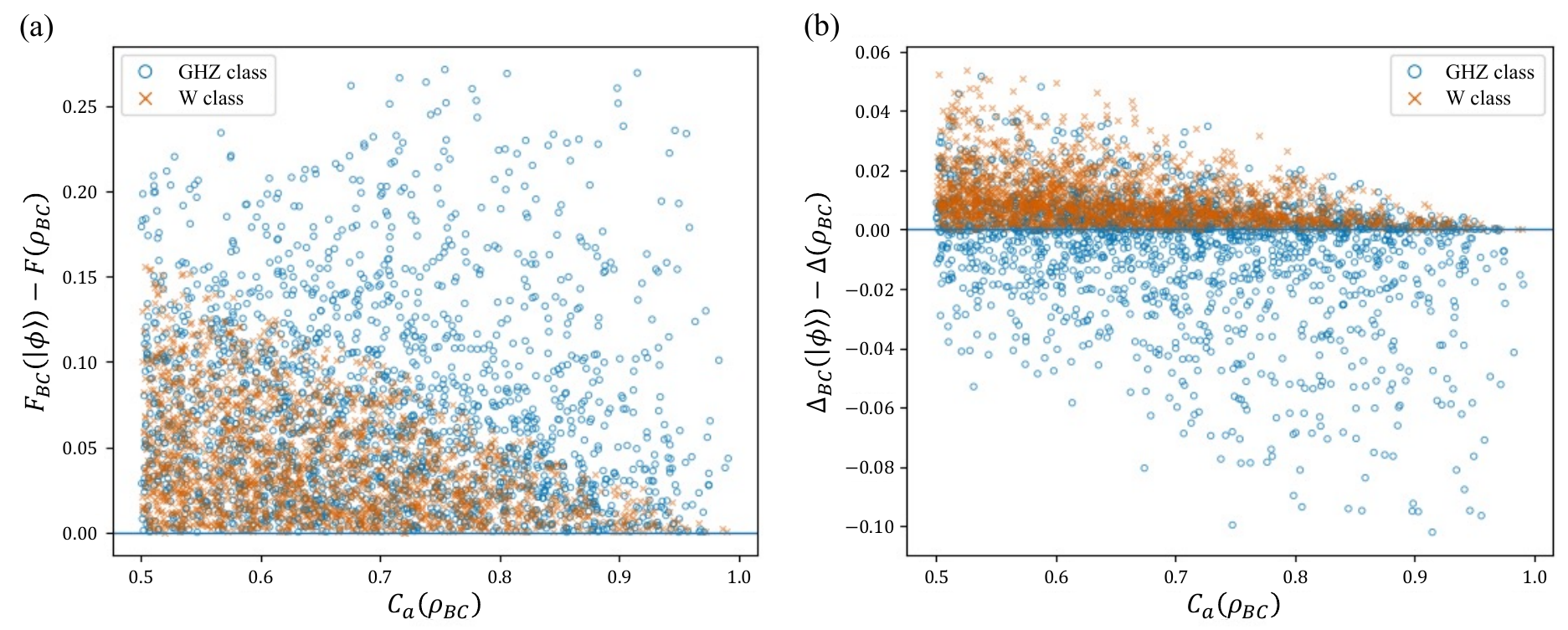}
\caption{
Effect of controller assistance on the maximal average fidelity and fidelity deviation for randomly sampled three-qubit pure states, with $A$ acting as the controller.
For each of the W and GHZ classes, 2000 randomly generated states satisfying $0.5 \le C_{a}(\rho_{BC}) \le 1$ are shown.
(a) Difference between the maximal average fidelity with controller assistance and that obtained directly from the reduced state $\rho_{BC}$.
(b) Difference between the minimal fidelity deviation with controller assistance and the fidelity deviation obtained directly from the reduced state $\rho_{BC}$.
The W-class samples satisfy $\Delta_{BC}(\ket{\phi}_{ABC})-\Delta(\rho_{BC}) \ge 0$, in accordance with Theorem~\ref{res:thm2}, whereas the GHZ-class samples exhibit both positive and negative differences.
}
\label{fig2}
\end{figure*}

For a two-qubit state $\varrho$ and the corresponding correlation matrix $R$ with singular values $s_1\ge s_2\ge s_3 \ge 0$, the maximal teleportation fidelity and the corresponding fidelity deviation are \cite{Horodecki1999,Ghosal2020}
\begin{equation}
F(\varrho)=\frac{3+\mathcal{N}(R)}{6}\quad\text{and}\quad\Delta(\varrho)=\sqrt{\frac{\Tr(R^{\mathsf T}R)}{30}-\frac{\mathcal{N}(R)^2}{90}},
\label{eq:mixed_fidelity_deviation}
\end{equation}
respectively,
where
\begin{equation}
\mathcal{N}(R)=\begin{cases}s_1+s_2+s_3, & \det R\leq0,\\s_1+s_2-s_3, & \det R>0.\end{cases}
\end{equation}
For $\rho_{BC}$, using Eq.~\eqref{eq:singular_order} together with $\det T_{BC}=\mu \det M = \mu^{2}(2\lambda_0^2-1)$ and
\begin{equation}
\quad \nu_+-\nu_-=\sqrt{\Tr(M^{\mathsf T}M)-2|\det M|},   
\end{equation}
we obtain
\begin{equation}
\mathcal{N}(T_{BC})=\mu+\sqrt{B},
\label{eq:w_N_B}
\end{equation}
where $B=\Tr(M^{\mathsf T}M)+2\mu(1-2\lambda_0^2)$.
Consequently,
\begin{align}
    F(\rho_{BC})&=\frac{3+\mu+\sqrt{B}}{6},
    \label{eq:w_reduced_fidelity}\\
    \Delta(\rho_{BC})&=\sqrt{\frac{\mu^2+\Tr(M^{\mathsf T}M)}{30}-\frac{(\mu+\sqrt{B})^2}{90}}.
    \label{eq:w_reduced_deviation}
\end{align}

We now show that
\begin{equation}
\label{eq:w_reduced_bound_squared}
\Delta^{2}(\rho_{BC})\le\left(\frac{1-C(\rho_{BC})}{3\sqrt{5}}\right)^{2}.
\end{equation}
Since $C(\rho_{BC})=2\lambda_{2}\lambda_{3}=\mu$, Eq.~\eqref{eq:w_reduced_bound_squared} is equivalent to
\begin{equation}
\label{eq:w_reduced_bound_squared_equiv}
L\leq\mu\sqrt{B},
\end{equation}
where $L=\Tr(M^{\mathsf T}M)+\mu(2\lambda_0^2+1)-1$.
Defining $R=1+\mu-2\lambda_0^2$, we obtain
\begin{align}
L&=\mu R-4\lambda_0^2\left(\lambda_2^2+\lambda_3^2-\mu\right)\notag\\&=\mu R-4\lambda_0^2(\lambda_2-\lambda_3)^2\le\mu R
\end{align}
and
\begin{equation}
B=R^2+4\lambda_0^2\lambda_1^2\geq R^2.
\end{equation}
Therefore, Eq.~\eqref{eq:w_reduced_bound_squared_equiv} is satisfied, completing the proof.
\end{proof}

The effect of control on the fidelity deviation can be qualitatively different for states outside the W class.
To illustrate this, let $\varphi=(1+\sqrt{5})/2$ and $\beta=\pi/12$, and consider the two states
\begin{align}
&\ket{\psi_W}_{ABC}=\frac{1}{2}(\ket{000}+\ket{100}+\ket{101}+\ket{110}), \nonumber \\
&\ket{\psi_{GHZ}}_{ABC}=\frac{\cos\beta}{\sqrt{\varphi}}\ket{000}+\frac{\cos\beta}{\varphi}\ket{100}+\sin\beta\ket{111}.
\end{align}
The first state belongs to the W class, whereas the second belongs to the Greenberger-Horne-Zeilinger~(GHZ) class~\cite{Dur2000}.
With party $A$ as the controller, both states have $C_{a}(\rho_{BC})=1/2$, and hence yield the same maximal average fidelity $F_{BC}=5/6$ and minimal fidelity deviation $\Delta_{BC}=1/(6\sqrt{5})$.
Their reduced states also have the same maximal average teleportation fidelity, $F(\rho_{BC}^{W})=F(\rho_{BC}^{GHZ})=(7+\sqrt{5})/12$, whereas their fidelity deviations are different.
In particular,
\begin{equation}
\Delta(\rho_{BC}^{W})=\frac{\sqrt{5}-1}{6\sqrt{10}} < \Delta_{BC} < \Delta(\rho_{BC}^{GHZ})=\frac{\sqrt{5}-1}{12}.
\end{equation}
Thus, control increases the fidelity deviation for $\ket{\psi_W}_{ABC}$, in accordance with Theorem~\ref{res:thm2}, while it decreases the fidelity deviation for $\ket{\psi_{GHZ}}_{ABC}$.
This example shows that states with the same maximal average fidelity and minimal fidelity deviation, as well as the same reduced-state fidelity, can nevertheless exhibit qualitatively different effects of control on the fidelity deviation.

To further illustrate the effect of control, we randomly sample W- and GHZ-class pure states in their canonical forms.
For each class, the squared amplitudes are sampled from a Dirichlet distribution, with the phase in the GHZ-class canonical form sampled uniformly from $[0, \pi]$.
We retain states satisfying $0.5 \le C_{a}(\rho_{BC}) \le 1$, with $A$ acting as the controller.
Figure~\ref{fig2} shows the resulting control power and the difference in fidelity deviation with and without controller assistance.
As expected from Eq.~\eqref{eq:diff_fidel}, the control power is nonnegative for both classes.
In contrast, the difference in fidelity deviation exhibits qualitatively different behavior.
It remains nonnegative for the W-class states, consistent with Theorem~\ref{res:thm2}, whereas both positive and negative values occur for the sampled GHZ-class states.
These results further demonstrate that, unlike the maximal average fidelity, the fidelity deviation does not exhibit a universal direction of change under controller assistance.

\section{Conclusion} 
\label{sec: conclusion}

In this work, we have investigated the fidelity deviation in controlled teleportation with three-qubit pure states.
We have formulated the fidelity deviation associated with the controller's measurement by taking into account both the fidelity fluctuations within each measurement outcome and those among different outcomes.
For three-qubit pure states, we have derived the minimal fidelity deviation in terms of the concurrence of assistance of the reduced two-qubit state and have shown that it can be attained simultaneously with the maximal average fidelity.
However, even among measurements attaining the same maximal average fidelity, the fidelity deviation can depend on the choice of the controller's measurement.
Finally, by comparing controlled teleportation with direct teleportation using the reduced state, we have found that, while controller assistance cannot decrease the maximal average fidelity, it does not necessarily decrease the fidelity deviation.
In particular, for W-class pure states, we have proved that the minimal fidelity deviation with controller assistance is no smaller than that without assistance, whereas states outside the W class can exhibit the opposite behavior.

The class-dependent behavior found in this work raises the question of whether the effect of controller assistance on the fidelity deviation is related to the underlying multipartite entanglement structure.
In particular, it would be interesting to investigate whether the difference between the fidelity deviations with and without controller assistance can be quantitatively related to quantities characterizing three-qubit entanglement, such as the three-tangle, or to measures of genuine multipartite entanglement.
Such an analysis may clarify what aspects of multipartite entanglement are reflected in the fidelity deviation beyond the concurrence of assistance.
Another direction is to extend the present analysis to mixed states, higher-dimensional systems, and more general multipartite settings, as well as beyond the standard teleportation protocol.
Finally, it would be intriguing to explore the operational significance of fidelity deviation in teleportation tasks where both the average fidelity and its variation over input states are relevant.

\par\medskip
\par\medskip
\section*{ACKNOWLEDGMENTS}
This research was supported by the Korea Institute of Science and Technology Information (KISTI) R\&D (Grant No. K26L3M3C3).

\bibliography{reference}

\end{document}